\documentclass{article}
\usepackage{graphicx}

\usepackage{amsmath,amssymb,amsthm}
\usepackage[linesnumbered,ruled,vlined]{algorithm2e}
\usepackage{dsfont}
\usepackage[T1]{fontenc}

\usepackage{authblk}
\usepackage{geometry}
\newtheorem{theorem}{Theorem}
\newtheorem{lemma}{Lemma}

\theoremstyle{definition}
\newtheorem{dfn}{Definition}

\newcommand{\ie}{{\em i.e., }}
\newcommand{\eg}{{\em e.g., }}

\newcommand{\calA}{\mathcal{A}}
\newcommand{\calL}{\mathcal{L}}

\newcommand{\calM}{\mathcal{M}}

\newcommand{\pal}{\calA_{svl}}

\newcommand{\slot}{\mathtt{slot}}
\newcommand{\id}{\mathtt{ID}}
\newcommand{\groupid}{\mathtt{group}}
\newcommand{\leaderid}{\mathtt{leader}}
\newcommand{\level}{\mathtt{lv}}
\newcommand{\mode}{\mathtt{mode}}
\newcommand{\inport}{\mathtt{port}}
\newcommand{\last}{\mathtt{last}}
\newcommand{\pin}{\mathtt{pin}}
\newcommand{\pout}{\mathtt{pout}}

\newcommand{\leader}{L}
\newcommand{\zombie}{Z}
\newcommand{\board}{S}

\newcommand{\settled}{\mathit{settled}}
\newcommand{\unsettled}{\mathit{unsettled}}

\newcommand{\maxdegree}{\delta_{\mathrm{max}}}
\newcommand{\lmin}{\calL_{\mathrm{min}}}
\newcommand{\vlevel}{\mathit{VL}}

\newcommand{\idmax}{\mathit{id}_{\mathrm{max}}}
\newcommand{\aset}{R}
\newcommand{\lset}{\aset_{\leader}}

\newcommand{\zset}{\aset_{\zombie}}
\newcommand{\bset}{\aset_{\board}}
\newcommand{\sinit}{s_{\mathrm{init}}}
\newcommand{\nextc}{\mathit{next}}

\newcommand{\asetmax}{\aset_{\mathrm{max}}}
\newcommand{\amax}{a_{\mathrm{max}}}

\newcommand{\vst}{v_{\mathrm{st}}}

\begin{document}
\title{
Efficient Dispersion of Mobile Agents without Global Knowledge
\thanks{
This work was supported by JSPS KAKENHI Grant Numbers 19H04085, 19K11826, and 20H04140 and JST SICORP Grant Number JPMJSC1606.
}
}

\author[1]{Takahiro Shintaku}
\author[1]{Yuichi Sudo\thanks{Corresponding author:y-sudou[at]ist.osaka-u.ac.jp}}
\author[2]{Hirotsugu Kakugawa}
\author[1]{Toshimitsu Masuzawa}

\affil[1]{Osaka University, Japan}
\affil[2]{Ryukoku University, Japan}

\date{}

\maketitle              % typeset the header of the contribution

\begin{abstract}
We consider the dispersion problem for mobile agents.
Initially, $k$ agents are located at arbitrary nodes
in an undirected graph.
Agents can migrate from node to node via an edge in the graph synchronously. Our goal is to let the $k$ agents be located at different $k$ nodes with minimizing the number of steps before dispersion is completed and the working memory space used by the agents.  Kshemkalyani and Ali [ICDCN, 2019] present a fast and space-efficient dispersion algorithm with the assumption that each agent has global knowledge such as the number of edges and the maximum degree of a graph. In this paper, we present a dispersion algorithm that does not require such global knowledge but keeps the asymptotically same running time and slightly smaller memory space.
\end{abstract}

\section{Introduction}
\label{sec:introduction}
%Mobile agents (or just \emph{agents})
%are abstract entity that autonomously
%migrates between nodes in a graph. 
%The agnet can represent software agents that  migrates
%from a computer to another in a computer network
%and/or  mobile robots that 

We consider the dispersion problem of mobile robots,
which we call mobile agents (or just \emph{agents})
in this paper.
At the beginning of an execution, 
$k$ agents are arbitrarily placed 
in an undirected graph at the beginning of an execution.
The goal of this problem
is to let all agents 
be located at different nodes. 
This problem was originally
formulated by Augustine and Moses Jr.~\cite{AM18} in 2018.
The most interesting point of this problem is 
the uniqueness of the computation model. 
Unlike many problems of mobile agents on graphs,
we cannot access the identifiers of the nodes 
and cannot use a local memory at each node,
usually called whiteboard.
In this setting, an agent cannot get/store any information
from/on a node when it visits the node. 
Instead, $k$ agents have unique identifiers
and can communicate with each other
when they visit the same node in a graph.
The agents must coordinately solve a common task
via direct communication with each other. 

\begin{table}[t]
\caption{Dispersion Algorithm for Arbitrary Undirected Graphs. 
($m'=\min(m,k\maxdegree/2,\binom{k}{2}$)}
\label{tbl:results}
\centering
\begin{tabular}{c c c c c}
\hline 
& \ \ \ Memory Space\ \ \ &\ \ \  Running Time\ \ \ &\ \ \ Knowledge\\
\hline
\cite{AM18}  & $O(k \log (\maxdegree+k))$ & $O(m')$ &\\
\cite{KA19} & $O(\ell \log (\maxdegree+k))$ bits & $O(m')$ steps &\\
\cite{KA19} & $O(d \log \maxdegree)$ bits & $O(\maxdegree^d)$ steps &\\
\cite{KA19}  & $O(\log (\maxdegree+k))$ bits & $O(m'\ell)$ steps &\\
\cite{KMS18} & $O(\log n)$ bits & $O(m'\log \ell)$ steps & $m$, $k$, $\maxdegree$\\
this work \ \ &  $O(\log (\maxdegree+k))$ bits & $O(m' \log \ell)$ steps & \\
\hline
\end{tabular}
\end{table}

Several algorithms were presented for the dispersion problem
of mobile agents in the literature.
Augustine and Moses Jr.~\cite{AM18} gave an algorithm 
that achieves dispersion within $O(m')$ steps,
where $m$ and $\maxdegree$ are the number of edges 
and the maximum degree of a graph, respectively,
and $m'=\min(m,k\maxdegree/2,\binom{k}{2})$.
In this algorithm, each agent uses $O(k \log (\maxdegree + k ))$ bits
of memory space. 
Kshemkalyani and Ali~\cite{KA19} presents
three dispersion algorithms.
The first one slightly decreases the memory space per agent
from $O(k \log (\maxdegree + k ))$ bits
to $O(\ell \log (\maxdegree + k ))$ bits,
where $\ell$ is the number of nodes at which
at least one agent is located at the beginning of an execution.
The second one uses $O(d \log \maxdegree)$ bits per agent
and solves the dispersion problem within $O(\maxdegree^d)$ steps,
where $d$ is the diameter of a graph.
The third one achieves dispersion
with much smaller space: $O(\log(\maxdegree + k))$ bits per agent.
However, it requires $O(m'\ell)$ time steps
before dispersion is achieved. 
The current state of art algorithm was given
by Kshemkalyani, Molla, and Sharma~\cite{KMS18}.
This algorithm is both time and space efficient.
The running time is $O(m' \log \ell)$ steps
and the memory space used by each agent
is $O(\log n)$ bits,
where $n$ is the number of nodes in a graph.
However, this algorithm requires global knowledge,
\ie $m$, $k$, and $\maxdegree$,
as mentioned in their paper.~\cite{KMS18}.
To the best of our understanding, 
their algorithm does not necessarily require the \emph{exact} values
of those parameters.
It requires only upper bounds
$M$, $K$, and $\Delta$ on $m$, $k$, and $\maxdegree$,
respectively.
Then, the assumption becomes much weaker,
but time and space complexities may increase
depending on how large those upper bounds are.
Indeed, given those upper bounds,
the algorithm achieves dispersion
within $O(M'\log \ell)$ steps
and uses $O(\log M')$ bits of memory space,
where $M=\min(M,K\Delta,K^2)$.
%Therefore, 
%If those upper bounds are asymptoticaly tight,

\subsection{Our Contribution}
The main contribution of this paper
is removing the requirement of global knowledge
of the algorithm given by \cite{KMS18}.
Specifically, we gave a dispersion algorithm
whose running time is $O(m' \log \ell)$ steps
and uses $O(\log (\maxdegree + k))$ bits of
the memory space of each agent. 
The proposed algorithm does not require
any global knowledge such as $m$, $k$, and $\maxdegree$.
In addition, the space complexity is slightly smaller
than the algorithm given by \cite{KMS18},
while the running times of both algorithms are asymptotically the same.

As with the existing algorithms listed in Table \ref{tbl:results},
the proposed algorithm works on an arbitrary simple, connected,
and undirected graph. 
This algorithm solves the dispersion problem
regardless of the initial locations of the agent in a graph.
We require that the agents are synchronous,
as in the algorithm given by \cite{KMS18}.

\subsection{Other Related Work}
The dispersion problem has been studied
not only for arbitrary undirected graphs,
but for graphs of restricted topology. %, such as trees and grids.
Augustine and Moses Jr.~\cite{AM18}
addressed this problem for paths, rings, and trees.
Kshemkalyani, Molla, and Sharma~\cite{KMS18}
studied the dispersion problem also in grid networks. 
Very recently, the same authors introduced 
\emph{the global communication model}~\cite{KMS20}.
Unlike the above setting, all agents can always communicate 
each other regardless of their current locations.
They studied the dispersion problem
for arbitrary graphs and trees under this communication model.

The exploration problem of mobile agents
is closely related to the dispersion problem.
This problem requires that each node (or each edge)
of a graph is visited at least once by an agent. 
If the unique node-identifiers are available,
a single agent can easily visit all nodes within $2m$ steps
in a simple depth first search traversal.
Panaite and Pelc \cite{PP99} gave a faster algorithm,
whose cover time is $m + 3n$ steps. 
%They assume that the nodes are labeled by the unique identifiers.
Their algorithm uses $O(m \log n)$ bits in the agent-memory,
while it does not use whiteboards, \ie local memories of the nodes.
Sudo, Baba, Nakamura, Ooshita, Kakugawa, and Masuzawa \cite{SBN15}
gave another implementation of this algorithm:
they removed the assumption of the unique identifiers
and reduced the space complexity on the agent-memory
from $O(m\log n)$ bits to $O(n)$ bits by using $O(n)$ bits in
each whiteboard.
The algorithm given by Priezzhev, Dhar, Dhar, and Krishnamurthy \cite{PDD+96}, which is well known as the \emph{rotor-router},
also solves the exploration problem efficiently.
The agent uses $O(\log \delta_v)$ bits
in the whiteboard of each node $v \in V$
and the agent itself is oblivious,
\ie it does not use its memory space at all.
%The edges $(\{v,u\})_{u \in N(v)}$
%are locally labeled by $0,1,\dots,\delta_v-1$ in a node $v$.
%The whiteboard of each node $v$ has one variable
%$v.\last \in \{0,1,\dots,\delta_v-1\}$.
%Every time it visits a node $v$,
%it increases $v.\lastport$ by one modulo $\delta_v$
%and moves to the next node via the edges
%labeled the updated value of $v.\lastport$.
The rotor-router algorithm is self-stabilizing,
\ie it guarantees that
starting from any (possibly corrupted) configuration, 
the agent visits all nodes within $O(mD)$ steps \cite{YWI+03}.

\subsection{Organization}
In Section \ref{sec:pre},
we define the model of computation and 
the problem specification.
In Section \ref{sec:existing},
we briefly explain the existing techniques used for the dispersion problem in the literature. This section may help the readers
to clarify what difficulties we addresses 
and how novel techniques we introduce to design the proposed protocol
in this paper.
In Section \ref{sec:svl}, 
we present the proposed protocol.
In section \ref{sec:conclusion},
we conclude this paper with short discussion for an open problem.

\section{Preliminaries}
\label{sec:pre}
Let $G=(V,E)$ be any simple, undirected, and connected graph.
Define $n=|V|$ and $m=|E|$. 
We define the degree of a node $v$ as $\delta_v = |\{u \in V \mid (u,v) \in E\}|$. Define $\maxdegree = \max_{v \in V} \delta_v$,
\ie $\maxdegree$ is the maximum degree of $G$.
The nodes are anonymous, \ie they do not have unique identifiers. However, the edges incident to a node $v$ 
are locally labeled at $v$
so that a robot located at $v$ can distinguish those edges.
Specifically, those edges have distinct labels
$0,1,\dots,\delta_v-1$ at node $v$.
We call this local labels \emph{port numbers}. 
We denote the port number assigned at $v$ for edge $\{v,u\}$
by $p_v(u)$.
Each edge $\{v,u\}$ has two endpoints, thus has labels $p_{u}(v)$ and $p_{v}(u)$.
Note that these two labels are not necessarily the same, \ie $p_{u}(v) \neq p_{v}(u)$ may hold.
We say that an agent moves \emph{via port $p$} from a node $v$
when the agent moves from $v$ to the node $u$ %via edge $\{v,u\}$
such that $p_{v}(u) = p$.

We consider that $k$ agents exist in graph $G$, where $k \le n$.
The set of all agents is denoted by
$\aset$.
Each agent is always located at some node in $G$,
\ie the move of an agent is \emph{atomic}
and an agent is never located at an edge at any time step
(or just \emph{step}).
The agents have unique identifiers,
\ie each agent $a$ has a positive identifier $a.\id$
such that $a.\id \neq b.\id$
for any $b \in \aset \setminus \{a\}$.
The agents knows a common upper bound $\idmax = O(poly(k))$
such that $\idmax \ge \max_{a \in \aset}a.\id$,
thus the agents can store the identifier of any agent
on $O(\log k)$ space.
%At each step, an agent can communicate with all agents
%located at the same node, while
%an agent cannot communicate with the agents
%located at diffrent nodes.
Each agent has a read-only variable $a.\pin \in \{-1,0,1,\dots,\delta_v-1\}$.
At time step $0$, $a.\pin = -1$ holds.
For any $t \ge 1$,
If $a$ moves from $u$ to $v$ at step $t-1$,
$a.\pin$ is set to $p_v(u)$
at the beginning of step $t$.
If $a$ does not move at step $t-1$,
$a.\pin$ is set to $-1$.

The agents are synchronous.
All $k$ agents are given a common algorithm $\calA$.
%Each step $t \ge 0$ consists of the following two steps:
Let $\aset(v,t) \subseteq \aset$ be the set of the agents located 
at a node $v$.
We define $\ell = |\{v \in V \mid \aset(v,0) \ge 1\}|$,
\ie $\ell$ is the number of nodes with at least one agent in step $0$.
At each step $t \ge 0$,
the agents in $\aset(v,t)$ first communicate with each other
and agrees how each agent $a \in \aset(v,t)$
updates the variables in its memory space in step $t$,
including a variable
$a.\pout \in \{-1,0,1,\dots,\delta_v-1\}$,
according to algorithm $\calA$.
The agents next update the variables according to the agreement.
Finally, 
each agent $a \in \aset(v,t)$ with $a.\pout\neq -1$
moves via port $a.\pout$.
If $a.\pout = -1$, agent $a$ does not move
and stays in $v$ in step $t$.

A node does not have any local memory
accessible by the agents. 
Thus, the agents have to coordinate only by
communicating each other. 
%It is worthwhile to mention that
No agents are given a priori any global knowledge
such as $m$, $\maxdegree$, and $k$.
 
The values of all variables in agent $a$
constitute the state of $a$.
Let $\calM_{\calA}$ be the set of all possible agent-states
for algorithm $\calA$. ($\calM$ may be an infinite set.)
Algorithm must specify the initial state $\sinit$,
which is common to all agents in $\aset$.
A global state of the network or a \emph{configuration}
is defined as a function $C:\aset \to (\calM,V)$
that specifies the state and the location of each agent $a \in \aset$.
%We write $C \to C'$ when 
In this paper, we consider only deterministic algorithms.
Thus, if the network is in a configuration $C$ at a step, 
a configuration $C'$ in the next step is uniquely determined.
We denote this configuration $C'$ by $\nextc_{\calA}(C)$.
%given a configuration $C_0$,
The execution $\Xi_\calA(C_0)$ of $\calA$
starting from a configuration $C_0$ is defined
as an infinite sequence $C_0,C_1,\dots$ of configurations
such that $C_{t+1} = \nextc_{\calA}(C_t)$ for all $t=0,1,\dots$.
%We say that a configuration $C'$ is
%\emph{recheable} from a configuration
%$C$ if $C'$ appears in $\Xi_\calA(C)$.
%By definition, $C$ is reachable from $C$ itself.
%where $C_t$ represents the configuration at step $t$.
%We say that a configuration $C$ is \emph{terminal}
%if $C = \nextc_{\calA}(C)$.

\begin{dfn}[Dispersion Problem]
A configuration $C$ of an algorithm $\calA$
is called \emph{legitimate} if
(i) all agents in $\aset$ are located in different nodes in $C$,
and (ii) no agent changes its location
in execution $\Xi_\calA(C)$. %all configuration reachable from $C$
We say that $\calA$ solves the dispersion problem
if execution $\Xi_\calA(C_0)$ reaches a legitimate configuration
for any configuration $C_0$ where all agents are in state $\sinit$.
\end{dfn}

We evaluate the \emph{running time} of algorithm $\calA$
as the maximum number of steps until $\Xi_\calA(C_0)$
reaches a legitimate configuration,
where the maximum is taken over $C_0$,
any configuration in which all agents are in state $\sinit$.

%In this paper, we frequently use 
We define $m'=\min(m,k\maxdegree/2,\binom{k}{2})$,
which is frequently used in this paper.
We simply write $\aset(v)$ for $\aset(v,t)$
when time step $t$ is clear from the context.

\section{Existing Techniques}
\label{sec:existing}

\begin{algorithm}[t]
\caption{Simple DFS}
\label{al:dfs}
\DontPrintSemicolon
\tcc{
the actions of all agents in $\aset(v)$ for any $v \in V$
}
$a.\pout \gets -1$ for all $a \in \aset(v)$
\tcp*{Initialize $\pout$.}
\If{$\aset(v) \ge 2$}{
\uIf{there is no settled agent in $\aset(v)$}{
% Pick an arbitrary agent %$b \in \aset(v)$\;
% and make it settled.\;
 The agent with the smallest identifier in $\aset(v)$
 becomes settled.\;
% \tcp*{The agent becomes $a_v$}
% $a_v.\mode \gets \settled$;
 $a_v.\last \gets a_v.\pin + 1 \bmod \delta_v$\;
 $a.\pout \gets a_v.\last$
 for all $a \in \aset(v) \setminus \{b\}$\;
}
\Else{
%Let $b$ be the (unique) setteled agent in $\aset(v)$.\;
Let $p$ be the common $a.\pin$ for $a \in \aset(v) \setminus \{a_v\}$\;
\uIf{$a.\last \neq p$}{
 $a.\pout \gets p$
 for all $a \in \aset(v) \setminus \{a_v\}$
 \tcp*{Backtrack}
}
\Else{
 $a_v.\last \gets a_v.\last + 1 \bmod \delta_v$\;
 $a.\pout \gets a_v.\last$
 for all $a \in \aset(v) \setminus \{a_v\}$\;
}
}
}
 %Let $b$ be the unique settled agent in $\aset(v)$\;
Every agent $a \in \aset(v)$ such that $a.\pout \neq -1$ moves via $a.\pout$\;
\end{algorithm}

\subsection{Simple DFS}
\label{sec:dfs}
If we assume that all $k$ agents are initially located at the same node $\vst \in V$,
the dispersion problem can be solved by
a simple depth-first search (DFS).
The pseudocode of the simple DFS is shown in
Algorithm \ref{al:dfs}.
Each agent $a \in \aset$ maintains a variable
$\mode \in \{\settled,\unsettled\}$.
We say that agent $a$ is \emph{settled}
if $a.\mode = \settled$,
and \emph{unsettled} otherwise.
All agents are initially unsettled.
Once an agent becomes settled at a node,
it never becomes unsettled
nor moves to another node.
%We call a node at which a settled agent is located
%\emph{a settled node}.
We say that a node $v$ is \emph{settled} if
an settled agent is located at $v$;
otherwise, $v$ is unsettled. 
We denote by $a_v$ a settled agent located at $v$
when $v$ is settled.
(No two agents become settled at the same node.)
In this model, no node provides
a local memory accessible by agents. %, so called a whiteboard.
However, when a node $v$ is settled,
unsettled agents can use the memory of $a_v$
like the local memory of $v$
since the unsettled agents
and $a_v$ can communicate with each other at $v$.

%Unsettled agents can visit $k$ nodes of a graph 
%in a DFS fashion by using the memory of settled agents.
Unsettled agents always move together in a DFS fashion. 
Each time they find an unsettled node $v$,
the agent with the smallest identifier
becomes settled at $v$ (Line 4).
An settled agent $a_v$ maintains a variable
$a_v.\last \in \{0,1,\dots,\delta_v-1\}$
to remember which port
%is used at $v$ for the last time
%by unsettled agents to leave $v$.
was used for the last time
by unsettled agents to leave $v$. %(Lines 5,6,12, and 13). 
At step 0, all agents are unsettled and located at node $\vst$.
First, after one agent becomes settled at $\vst$,
the other $k-1$ agents move via port 0.
Thereafter, 
the unsettled agents basically migrates between nodes
by the following simple rule:
each time they move from $u$ to $v$,
it moves via port $p_{v}(u)+1 \bmod \delta_v$
(Lines 5, 6, 12, and 13).
Only exception is the case that
node $v$ has already been settled when they move from $u$ to $v$
and $p_{v}(u) \neq a_v.\last$.
At this time, the unsettled agents immediately backtracks
from $v$ to $u$
and this backtracking does not update $a_v.\last$ (Line 10).

If $k > n$,
this well-known DFS traversal guarantees that
the unsettled agents visit all nodes within $4m$ steps,
during which the agents move through each edge at most four times.
Since we assume $k \le n$,
the unsettled agents move through at most
$m'=\min(m,k\maxdegree/2,\binom{k}{2})$ 
different edges.
Thus, all agents become settled \footnotemark{}
within $4m'=O(m')$ steps,
at which point the dispersion is achieved.
The space complexity is $O(\log \maxdegree)$ bits per agent:
each agent maintains only one non-constant variable $\last$,
which requires $\lceil \log \maxdegree \rceil$ bits.
\footnotetext{
Strictly speaking, according to Algorithm \ref{al:dfs},
the last one agent never becomes settled
even if it visits an unsettled node. 
However, this does not matter
because thereafter the last agent never moves 
nor change its state.
}

\subsection{Parallel DFS}
\label{sec:parallel}
Kshemkalyani and Ali \cite{KA19} generalized the above simple DFS
to handle the case that the agents may be initially located
at multiple nodes, using $O(\ell \log (\delta+k))$ bits per agent. (Remember that $\ell$ is the number of nodes at which one or more agents are located in step 0.)
That is, $\ell$ groups of unsettled agents
perform DFS in parallel.
Specifically, in step 0, the agents located at each node $v$
compute $\max\{a.\id \mid a \in A(v,0)\}$
and store it in variable $\groupid$.
Thereafter, the agents use this value as \emph{group identifier}.
%We call a set of agents with the same group identifier
%a \emph{group}.
%Unsettled agents of each group requires settle agents 
%Even when two or more groups of unsettle agents may
%visit the same node, each agent can 
Settled agents can distinguish each group by group identifiers,
thus they can maintain $\ell$ slots of individual memory space
such that unsettled agents of each group can dominantly
access one slot of the space.
Since simple DFS requires $O(\log \maxdegree)$ bits,
this implementation of parallel DFS requires 
$O(\ell \log \maxdegree + \ell \log k))=O(\ell \log (\maxdegree+k))$ bits. The running time is still $O(m')$ steps.
%Augustine and Moses Jr.~\cite{AM18} implemented Parallel DFS
%in another way where each unsettled agent $a$ memorizes
%the identifiers of the setteld agents that $a$ has met before.
%The running time of this implementation is also $O(m')$,
%but requires a slightly larger memory space,
%\ie $O(k \log (\maxdegree+k)$ bits per agent. 

\subsection{Zombie Algorithm}
\label{sec:zombie}
We can solve the dispersion
with memory space of $O(\log \maxdegree)$ bits per agent at the cost of increasing the running time to $O(m'k)$,
regardless the initial locations of $k$ agents.
%Each settled agent provides only one group with memory space of $O(\log \maxdegree)$ bits.
In step 0, all agents compute its group identifier
in the same way as in Section \ref{sec:parallel}.
However, due to the memory constraints,
each settled agent cannot maintain one slot of memory space
for each of $\ell$ groups.
Instead, it provides only one group with memory space of $O(\log \maxdegree)$ bits.
Each settled agent memorizes the largest group identifier
it observes: each time unsettled agent visits a node $v$,
$a_v.\groupid$ is updated to $\max\{a.\groupid \mid a \in \aset(v)\}$.
Thus, at least one group of agents can perform its DFS
by using the memory space of settled agents exclusively.
If an unsettle agent $a$ visits a node $v$ such that 
$a.\groupid < a_v.\groupid$,
$a$ becomes a \emph{zombie},
which always chase the agent whose identifier is equal to
$a_v.\groupid$, which we call the \emph{leader} of the group.
Specifically, 
a zombie $z$ chases a leader by moving via port $a_u.\last$
each time $z$ visits any node $u$.
When $z$ catch up with the leader,
$z$ always follow the leader thereafter until
$z$ becomes settled, which occurs
when $z$ reaches an unsettled node $v$ and
$z$ has the smallest identifier among $A(v)$.
Every agent moves at most $4m'$ times until it observes larger group identifier than any identifier it has observed so far.
%while it changes its $\groupid$ at most $\ell-1$ times. 
%the number of group identifiers is $\ell$.
Thus, the running time of this algorithm is $O(m'\ell)$ steps.

We call this algorithm the \emph{zombie algorithm}.
The zombie algorithm is essentially the same to
the tree-switching algorithm given by \cite{KA19},
while the zombie algorithm is faster only by a constant factor.

\subsection{Dispersion with  {\boldmath $O(m'\log \ell)$} steps
and {\boldmath $O(\log n)$} bits per agent}
\label{sec:kms19}

Kshemkalyani, Molla, and Sharma \cite{KMS18}
gave the current state of art algorithm
that achieves dispersion
%within $O(m'\log \ell)$ steps 
%using only $O(\log n)$ bits
on arbitrary graphs regardless of the initial locations of the $k$ agents.
%Their algorithm requires the knolwedge of $n$, $m$, $k$, and $\maxdegree$. 
%However, it does not require \emph{exact} values
%of those parameter.
%Indeed, it is enough if the agents know upper bounds on
%those parameters, say $N$, $M$, $K$, and $\Delta$, respectively.
Their algorithm requires that 
each agent knows \emph{a priori} upper bounds $M$, $K$, and $\Delta$
on $m$, $k$, and $\maxdegree$,
respectively.
Then, the running time is $O(\min(M,K\Delta,K^2) \log \ell)$
steps and the required memory space per agent
is $O(\log M)$ bits.
In particular, if those upper bounds are asymptotically tight,
\ie $N=O(n)$, $M=O(m)$, $K=O(k)$, and $\Delta=O(\maxdegree)$,
those complexities are $O(m' \log \ell)$ steps and
$O(\log n)$ bits, respectively.
In this section, we briefly explain the key idea of their algorithm to clarify how this algorithm requires global knowledge
$M$, $K$, and $\Delta$.
The implementation in the following explanation
is slightly different from the original implementation in
\cite{KMS18}, but the difference is not essential.

%The key idea of their algorihtm is simple, but non-trivial and excellent. 
Let $T$ be a sufficiently large $O(\min(M,K\Delta,K^2))$ value. Each unsettled agent maintains a timer variable and
counts how many steps have passed since an execution began
modulo $T$. %Agents are synchronously 
Agents switches from stage 1 to stage 2
and from stage 2 to stage 1 in every $T$ steps.
At the switch from stage 1 to stage 2
and from stage 2 to stage 1,
all settled agents reset their group identifiers 
to $-1$, which is smaller than any group identifier
of unsettled agents.

In stage 1, all unsettled agents at each node $v$
first compare their identifiers and
adopts the largest one as their group identifier.
%each group of unsettled agents peforms DFS.
%like the algorithm explained in Section \ref{sec:parallel}.
Then, unsettled agents perform DFS in parallel
like the zombie algorithm.
The difference arises when an unsettled agent
finds a settled agent with a larger group identifier.
Then, instead of becoming a zombie,
it stops until the end of stage 1.
Moreover, unlike the zombie algorithm,
an unsettled agent located at a node $v$
becomes settled even when
$|\aset(v)| = 1$.

At the beginning of stage 2, 
%there may be still unsettled agents,
%which are located at settled agents.
one or more unsettled agents
may be located at settled nodes.
%Let $a$ be any such agent and $v$ the agent at which $v$ is located.
%Let $G'(a)$ be the connected component that includes $v$
%in the subgraph induced by all settled nodes in the beginning of
%stage 2.
Let $G'$ be the subgraph induced by 
all settled nodes in the beginning of stage 2
and $G'(v)$ the component that includes a node $v$ in $G'$.
The goal of stage 2 is to collect all agents in $G'(v)$
and locate them to the same node in $G'(v)$ for each $v \in V$.
Specifically, each unsettled agents located at any node $u$
tries to perform DFS in $G'(u)$ twice.
Since each settled agent $b$ has only $O(\log M)$ bits space,
the memory space of $b$ can be used only by unsettled agents
in the group with the largest identifier that $b$ has observed
in stage 2. An settled agents with smaller
group identifier stops when it observes a larger group identifier.
Then, all unsettled agents in the group with the largest group identifier in $G'(u)$ can perform DFS and visit all nodes in $G(u')$
twice within $T$ steps. 
During the period, they pick up all stopped and unsettled agents in the component and goes back to the node that they are located at in the beginning of stage 2.

Hence, each iteration of stages 1 and 2
decreases the number of nodes with at least one unsettled agent
at least by half.
This yields that all agents become settled 
and the dispersion is achieved
within $O(T \log \ell)=O(\min(M,K\Delta,K^2) \log \ell)$ steps.

The knowledge of the global knowledge
$M$, $K$, and $\Delta$ is inherent to the key idea of this algorithm. Without those upper bounds,
%the agents cannot switch between two stages in an appropriate number of steps.
the agents may switch between two stages
before all agents complete a stage so that
the correctness is no longer guaranteed
or the agents may stay in one stage too long
so that the running time is much larger than, for example,
 $\Omega(mn)$.
To the best of our knowledge,
no simple modification removes the requirement
of the knowledge keeping the same running time asymptotically.

\section{Proposed Algorithm}
\label{sec:svl}

In this section, we give an algorithm $\pal$
that solves a dispersion within $O(m' \log \ell)$ steps
and uses $O(\log (k+\maxdegree))$ bits of memory space
per agent. Algorithm $\pal$ requires no global knowledge.

\subsection{Overview}
\label{sec:overview}
Algorithm $\pal$ is based on the zombie algorithm
explained in Section \ref{sec:zombie}
but it has more sophisticated mechanism
to achieve dispersion within $O(m' \log \ell)$ steps.
An agent $a$ maintains its \emph{mode}
on a variable $a.\mode \in \{\leader,\zombie,\board\}$. %to represent the current mode of $a$.
We say that an agent $a$ is
a leader (resp.~a zombie, a settled agent)
when $a.\mode = \leader$ (resp.~$a.\mode = \zombie$, $a.\mode = \board$).
A leader may become a zombie,
and a zombie eventually becomes a settled agent.
However, a zombie never becomes a leader again,
and a settled agent never changes its mode.
An agent $a$ also maintains its \emph{level}
on a variable $a.\level$,
while $a$ stores the identifier of some leader on a variable
$a.\leaderid$.
%A leader $a$ must alwas hold $a.\id = a.\leader$.
As long as an agent $l$ is a leader, $l.\id = l.\leaderid$ holds.
These two variables, $\level$ and $\leaderid$,
determine the strength of agent $a$.
We say that an agent $a$ is stronger than an agent $b$
if $a.\level > b.\level$ or
$a.\level = b.\level \land a.\leaderid > b.\leaderid$
holds.

%In an execution of $\pal$,
All agents are leaders
at the beginning of an execution of $\pal$,
Each time two or more leaders visit the same node,
the strongest leader \emph{kills} all other leaders,
and the killed leaders become zombies.
%a stronger leader kills a weker leaders when they visit the same ndoe, and the killed leader becomes a zombie.
Zombies always follow the strongest leader that
it has ever observed. 
A leader always tries to perform a DFS.
Each time a leader visits an unsettled node $v$,
%and at least one zombie is located at $v$,
the leader picks an arbitrary one zombie $z \in \aset(v)$
and makes $z$ settled.
%if at least one zombie is located at $v$.
If there is no zombie at $v$, 
the leader suspends its DFS until zombies visits $v$.
As in Section \ref{sec:existing},
we denote the settled node located at $v$ 
by $a_v$ if it exists.
We say that a settled agent $s$
is a \emph{minion} of a leader $l$
when $s.\level = l.\level$ and $s.\leaderid = l.\id$ hold.
As in the zombie algorithm, 
a leader $l$ performs a DFS
by using the memory space of its minions,
in particular, by using the values of their $\last$s.
%Specifically, a leader $l$ moves by the following rule:
%\begin{itemize}
% \item when $l$ moves from $u$ to $v$ 
%\end{itemize}
%When a leader $l$ visits a settled node $v$
%such that the settled agent $a_v$ is weaker than $l$,
When a leader $l$ visits a settled node $v$
such that $a_v$ is weaker than $l$,
$a_v$ becomes a minion of $l$
by executing  $(a_v.\level,a_v.\leaderid) \gets (l.\level,l.\id)$.
Conversely, 
the leader $l$ becomes a zombie if $a_v$ is stronger than $l$.
At this time, the stronger leader may not be located at node $v$. However, $a_v$ can chase the stronger leader by moving via $a_v.\last$ repeatedly.

%As mentioned abeove,
%when two or more leaders visit the same node,
%only the strongest leader remains the leader.
Each leader maintains its level like
the famous minimum spanning tree algorithm given by
Gallager, Humblet, and Spira \cite{GHS83}.
Initially, the level of every agent is zero.
A leader increases its level by one each time
it kills another leader with the same level
or meets a zombie with the same level.
We have the following lemma.
\begin{lemma}
\label{lem:maxlevel} 
No agent reaches a level larger than $\log_2 \ell+1$.
\end{lemma}
\begin{proof}
Let $\ell_2$ be the nodes at which two or more leaders 
are located in the initial configuration,
\ie $\ell_2 = |\{v \in V \mid |\aset(v,0)| \ge 2 \}|$.
In step 0, (i) $\ell_2$ leaders changes their levels from 
0 to 1,
(ii) the levels of $\ell - \ell_2$ leaders remain 0, and 
(iii) other $k-\ell$ leaders become zombies with level 0.
Thus, at most $\ell_2+(\ell-\ell_2)=\ell$ leaders 
reaches level 1. 
Since a zombie never increases its level
until it becomes a settled agent,
for any $i \ge 1$, at most $\ell/2^{i-1}$ leaders reaches
level $i$.
Thus, no agent reaches a level larger than $\log_2 \ell+1$. 
\end{proof}

By definition of minions, a leader $l$ loses all its minions
when $l$ increases its level by one.
As we will see in Section \ref{sec:detailed},
when a leader visits a node $v$,
the leader considers that it has already visited a node $v$
in its DFS if and only if $a_v$ is a minion at that time. 
This means that a leader $l$ restarts a new DFS
each time it increases its level.
%This property is helpful to prove the correctnes of
%algroithm $\pal$.
In every DFS, a leader $l$ moves at most $4m'$ times. 
However, this does not necessarily mean that $l$ completes its DFS
or increases its level within $4m'$ steps
because $l$ suspends its DFS while no other agent is located at the same node. 
%To achieve dispersion within $O(m' \log \ell)$ steps,
%it is desirable to guarantee that 
Thus, we require a mechanism to bound the running time 
by $O(m' \log \ell)$ steps.
As we will see in Sections \ref{sec:detailed}
and \ref{sec:analysis},
we achieve this by differentiating the moving speed of agents 
according to various conditions.

\begin{table}[t]
\caption{Variables of $\pal$}
\label{tbl:svl}
\centering
\begin{tabular}{r c l c}
\hline
\multicolumn{1}{c}{variables} & & \multicolumn{1}{c}{description} & initial value \\
\hline
$a.\mode \in \{\leader,\zombie,\board\}$
&:&
the mode of an agent $a$
%$\leader$, $\zombie$, and $\board$
%means a leader, a zombie, and a blackboard, respectively.
&  $\leader$
\\
$a.\slot \in \{0,1,2,3\}$
&:& the current timeslot
&  $0$
\\
$a.\level \in \mathbb{N}$ 
&:& the level of an agent $a$
& $0$
\\
$a.\leaderid \in \mathbb{N}$
&:& the identifier of the strongest leader that $a$ observed
& $a.\id$
\\
$a.\last \in \mathbb{N}$
&:&
%the last port that a leader used to leave the current node
the pointer to the strongest leader that $a$ has observed
& $0$
\\
$a.\inport \in \mathbb{N} \cup \{-1\}$ 
&:& the last non-negative value of $\pin$
& $-1$
\\
\hline
\end{tabular}
\end{table}

\begin{algorithm}
\caption{$\pal$}
\label{al:svl}
\DontPrintSemicolon
\tcc{
the actions of all agents in $\aset(v)$ for any $v \in V$
}
%Define agent $\amax$ as follows:
Let $\amax$ be the unique leader in $\asetmax(v)$ if it exists.
Otherwise, let $\amax = a_v$.\;
%{
% \begin{itemize}
%%  \item Let $\asetmax(v)$ be the set of strongest agents in $\aset(v)$.
%  \item Let $\amax$ be the unique leader in $\asetmax(v)$
%	if it exists.
%  \item Otherwise, let $\amax = a_v$.
%	%let $\amax$ be the agent with the largest identifier in $\asetmax(v)$.	
% \end{itemize}
%}
%\tcc{$|\lset(v) \cap \asetmax(v)| \le 1$ and $\amax.\mode \in \{\leader,\board\}$ hold here.}
$a.\pout \gets -1$ for all $a \in \aset(v)$
\tcp*{Initialize $\pout$.}
$a.\mode \gets \zombie$ for all $a \in \lset(v) \setminus \{\amax\}$ \tcp*{$\amax$ kills the other leaders.} 
 $
 \amax.\inport \gets 
 \begin{cases}
 \amax.\pin & \amax.\mode = \leader \land \amax.\pin \neq -1 \\
 \amax.\inport & \text{otherwise}
 \end{cases}
 $\;
\If{$|\aset(v)| \ge 2$}{
\uIf{$\amax.\mode = \leader$}{
 $
 \amax.\level \gets
 \begin{cases}
 \amax.\level + 1 & \exists z \in \zset(v): z.\level = \amax.\level \\
 \amax.\level & \text{otherwise}
 \end{cases}
 $\;
 \uIf{$\bset(v)=\emptyset$}{
 Choose any zombie
 in $\aset(v) \setminus \{\amax\}$
 and makes it 
 settled.\;
 $a_v.\level \gets \amax.\level$; $a_v.\leaderid \gets \amax.\leaderid$;  $a_v.\last \gets \amax.\inport$
 }
 \uElseIf{
 $a_v$ is not a minion of $\amax$
 }{
 $a_v.\level \gets \amax.\level$; $a_v.\leaderid \gets \amax.\leaderid$; $a_v.\last \gets \amax.\inport$
 %$a_v.\last \gets \amax.\pin$
 }
 \uElseIf{
$a_v$ is a minion of $\amax$ \textbf{\emph{and}}
$\amax.\inport \neq a_v.\last$}{
 $a.\pout \gets \amax.\inport$
 for all $a \in \aset(v) \setminus \{a_v\}$
 \tcp*{Backtrack}
 }
 \ElseIf{$\amax.\slot = 0$}{
 $a.\pout \gets \amax.\inport + 1 \bmod \delta_v$
 for all $a \in \aset(v) \setminus \{a_v\}$\;
 $a_v.\last \gets \amax.\pout$
 }
 }
 \ElseIf{$
\left(
\begin{aligned}
&\max_{z \in \zset(v)} z.\level < a_v.\level \land a_v.\slot \in
 \{2,3\}\\
& \lor
\max_{z \in \zset(v)} z.\level = a_v.\level \land a_v.\slot = 2
\end{aligned}
\right)$
}{
$a.\pout \gets a_v.\last$ for all $a \in \aset(v) \setminus \{a_v\}$\;
}
}
% \Else{
% $
% a.\pout \gets 
% \begin{cases}
%  a_v.\last  &
%\left(
%\begin{aligned}
%&\max_{z \in \zset(v)} z.\level < a_v.\level \land a_v.\slot \in
% \{2,3\}\\
%& \lor
%\max_{z \in \zset(v)} z.\level = a_v.\level \land a_v.\slot = 2
%\end{aligned}
%\right)
%\\
%  -1 & \text{otherwise}
% \end{cases}
% $
% for all $a \in \aset(v) \setminus \{a_v\}$\;
% }
%}  
 $a.\slot \gets a.\slot + 1 \bmod 4$ for all $a \in \aset(v)$\;
Every agent $a \in \aset(v)$ such that $a.\pout \neq -1$ moves via $a.\pout$\;
\end{algorithm}

\subsection{Detailed Description}
\label{sec:detailed}
The list of variables and the pseudocode of $\pal$
are given in Table \ref{tbl:svl} and Algorithm \ref{al:svl},
respectively.

First, we introduce some terminologies and notations.
Let $t$ be any time step.
We denote the set of leaders, the set of zombies,
and the set of settled agents at step $t$
by $\lset(t)$, 
$\zset(t)$, and $\bset(t)$, respectively.
In addition, we define
$\lset(v,t) = \aset(v,t) \cap \lset(t)$,
$\zset(v,t) = \aset(v,t) \cap \zset(t)$, and
$\bset(v,t) = \aset(v,t) \cap \bset(t)$
for any node $v \in V$.
Define $\asetmax(v,t)$ as
the set of strongest agents in $\aset(v,t)$.
We omit time step $t$ from those notation,
\eg simply write $\lset(v)$ for $\lset(v,t)$,
when time step $t$ is clear from the context.
A leader $l$ located at a node $v$ is 
called an \emph{active leader} when $|\aset(v)|\ge 2$.
When $|\aset(v)| = 1$,
$l$ is called a \emph{waiting} leader.
A zombie $z$ located at a node $v$
is called a \emph{strong zombie} 
if $z.\level = a_v.\level$.
Otherwise, $z$ is called a \emph{weak zombie}.
In the pseudocode,
we use notation $\amax$ for simplicity.
We define $\amax$ as follows\footnotemark{}:
If there is a leader in $\asetmax(v)$,
$\amax$ denotes the (unique) leader in $\asetmax(v)$;
Otherwise, %the agents in $\aset(v)$ picks up
%an arbitrary agent in $\asetmax(v)$ as $\amax$.
let $\amax = a_v$. 
\footnotetext{
This definition assumes that
there is at most one leader in $\asetmax(v)$
and there is no zombie in $\asetmax(v)$ at any time step.
%Both propositions trivially hold
The former proposition holds because
for any two leaders $l_1$ and $l_2$,
$l_1$ is stronger than $l_2$, or
$l_2$ is stronger than $l_1$.
%and 
The latter holds because
(i) a leader becomes a zombie only if
it observes a stronger agent,
(ii) a zombie never becomes stronger,
\ie never change its $\level$ or $\leaderid$,
unless it becomes settled,
and 
(iii) a zombie moves by chasing a stronger leader
or moves together with a stronger leader.
}
%Since $\asetmax(v) \subseteq \lset(v) \cup \bset(v)$,
%$\amax$ is either a leader or a settled agent.

%In addition to $\mode$, $\level$, $\leaderid$, and $\last$,
%we introduce two variables $\slot \in \{0,1,2,3\}$
%and $\inport \in \mathbb{N}$.
%Each agent $a$ counts how many steps
%have passed since an execution of $\pal$ began
%modulo $4$ and stores on $a.\slot$ (Line 20).
%Thus, all agents always have the same value 
%for this variable.
%%Each settled agent $s$ maintains $s.\last$
%Each leader $l$ remembers the last non-negative value
%of $l.\pin$ on a variable $l.\inport$,
%which is used for performing a DFS.

As mentioned in Section \ref{sec:overview},
we differentiate the moving speed of agents.
We implement the differentiation with
a variable $\slot \in \{0,1,2,3\}$.
Each agent $a$ counts how many steps
have passed since an execution of $\pal$ began
modulo $4$ and stores it on $a.\slot$ (Line 20).
Thus, all agents always have the same value 
for this variable.
This variable represents timeslots
indicating which kind of agents are allowed to move 
at each step.
For example, an active leader is allowed to move
at slots 0 and 1,
while a strong zombie is allowed to move only at slot 2.

The pseudocode (Algorithm \ref{al:svl})
specifies how the agents located at a node $v$
updates their variables including port $\pout$, 
via which they move to the next destination. 
A leader is immediately killed, \ie becomes a zombie,
if it meets a stronger leader or a stronger settled agent
(Line 3).
A surviving leader increases its level by one each time it meets
a (weaker) leader or zombie with the same level
(Line 7).
A leader $l$ performs a DFS using a variable $\last$
in its minions.
It basically moves in timeslot 0,
while backtracking occurs in timeslot 1.
Since the information of incoming port 
in $l.\pin$ gets lost in timeslots 2 and 3,
the leader $l$ remembers $\pin = p_{v}(u)$
in a variable $l.\inport$
each time $l$ moves from $u$ to $v$ (Line 4).
Specifically, a leader $l$ performs a DFS as follows
until it is killed by another leader:
\begin{itemize}
 \item When $l$ moves from $u$ to $v$ such that
$a_v$ is not a minion of $l$, the leader $l$ considers that 
it visits $v$ for the first time in the current DFS.
Then, $l$ changes $a_v$ to a minion of $l$ (Line 12)
and waits for the next timeslot 0.
 \item When $l$ moves from $u$ to $v$ such that
$v$ is an unsettled node,
$l$ considers that 
it visits $v$ for the first time in the current DFS.
If there is at least one zombie at $v$, 
$l$ changes arbitrary one zombie to a settled agent
and make it a minion of $l$ (Lines 9 and 10).
Otherwise, it suspends the current DFS until
a zombie or a (weaker) leader $a$ visits $v$
(the condition in Line 5 implements this suspension).
Then, $l$ does the same thing for $a$,
\ie make $a$ settled and a minion of $l$.
Thereafter, $l$ waits for timeslot 0.
 \item When $l$ moves from $u$ to $v$
such that $a_v$ is a minion and $a_v.\last = l.\inport$,
the leader $l$
considers that it has just backtracked from $u$ to $v$.
Then, it just waits for the next time slot 0.
%The leader does nothing in this time step.
 \item %When the above three cases occur,
In timeslot 0, 
$l$ leaves the current node $v$
via port $l.\inport + 1 = p_{v}(u) + 1$
after storing this port on $a_v.\last$
(Lines 16 and 17).
 \item When $l$ moves from $u$ to $v$
such that $a_v$ is a minion and $a_v.\last \neq l.\inport$, 
the leader $l$ considers that 
it has already visited $v$ before in the current DFS.
Thus, it immediately backtracks from $v$ to $u$
without updating $a_v.\last$ (Line 15).
Since a leader makes a non-backtracking move
only in timeslot 0, this backtracking occurs
only in timeslot 1.
\end{itemize}
In the first and second case, 
the leader $l$ simultaneously
substitute $l.\inport$ for $a_v.\last$
to avoid triggering backtracking mistakenly.
Thus, each leader can perform a DFS in the same way as
a simple DFS explained in Section \ref{sec:dfs}
until it is killed and becomes a zombie.
%A zombie simply chases 

A zombie always tries to follow a leader.
If a zombie $z$ is located at a node $v$ without a leader,
it moves via port $a_v.\last$ in timeslot 2
if $z$ is a strong zombie or some strong strong zombie
is located at the same node $v$ (Line 19).
Otherwise, $z$ moves via port $a_v.\last$
in timeslots 2 and 3 (Line 19).
If $z$ is located at a node with a leader,
it just moves with the leader (Lines 14 and 16)
or becomes settled (Line 9).

%Specifically, we divide all agents into the following four types.
%\begin{description}
% \item[Type I]: An active agent or a zombie located at a node where at least one node exists.
% \item[Type II]: A strong zombie or a weak zombie located at a node where at least one strong zombie exits. 
% \item[Type III]: 
% \item[Type IV]:
%\end{description}

%As mentioned above, 
%each leader performs a DFS
%by using variable $\last$ of 
%to store the pointer to the strongest leader
%that $s$ observed (or possibly a stronger leader).

\subsection{Correctness and Complexities}
\label{sec:analysis}
In this section, we prove that
an execution $\Xi_{\pal}(C)$ of $\pal$
starting from any configuration $C$
where all agents are in the initial state
achieves dispersion 
within $O(m' \log \ell)$ steps,
and each agent uses $O(\log (k+\maxdegree))$
bits of memory space.

First, we gave an upper bound on the space complexity.
\begin{lemma}
\label{lem:space} 
Each agent uses at most $O(\log (k+\maxdegree))$ bits
of its memory space in an execution of $\pal$.
\end{lemma}
\begin{proof}
Each agent $a$ maintains six non-constant variables:
$a.\level$, $a.\leaderid$, $a.\last$, $a.\inport$, 
$a.\pin$, and $a.\pout$.
The first variable $a.\level$ uses only $O(\log \log \ell)$ bits
by Lemma \ref{lem:maxlevel}.
Each of the other five variables just stores 
the identifier of some agent
or the port numbers of some node.
Thus, they uses only $O(\log (k + \maxdegree))$ bits.
\end{proof}

%By definition of $\pal$,
In an execution of $\pal$,
agents located at a node $v$
leaves $v$ only when there are two or more agents at $v$. 
Therefore, we have the following lemma.

\begin{lemma}
\label{lem:once}
In an execution of $\pal$,
once all agents are located at different nodes,
no agent leaves the current location thereafter.
\end{lemma}

Thus, it suffices to show that
all agents will be located at different nodes
within $O(m' \log n)$ steps
in an execution of $\pal$.
To prove this,
%we define the following potential function.
we introduce some terminologies and notations.
Let $t$ be any time step
and let $a$ be any agent in $\aset(v,t)$ for any $v$.
Then, define the \emph{virtual level} of agent $a$
at step $t$
as $\vlevel(a,t) = \max_{b \in \aset(v)} b.\level$.
Define $\lmin(t)$ be the minimum virtual level 
of all active leaders and all zombies in $\aset$
at step $t$.
For simplicity, we define $\lmin(t) = \infty$
if there is no active leader and no zombie
in $\aset$ at step $t$.
Again, we simply write $\vlevel(a)$ and $\lmin$ for
$\vlevel(a,t)$ and $\lmin(t)$, respectively,
when time step $t$ is clear from the context.
Note that for any agent $a$,
the virtual level and the level of $a$
differ if and only if $a$ is a weak zombie.

By definition of $\pal$, we have the following lemma.
\begin{lemma}
\label{lem:nondecreasing} 
In an execution of $\pal$,
for any agent $a$,
the virtual level of $a$ never decreases.
\end{lemma}

\begin{lemma}
\label{lem:leader_dfs}
In an execution of $\pal$,
from any time step,
each active leader $l$ increases its level at least by one
or becomes a waiting leader, a zombie, or a settled agent
within $O(m')$ steps.
\end{lemma}

\begin{proof}
The active leader $l$ performs a DFS correctly
unless it meets a leader/zombie with the same level
or finds a stronger agent than $l$.
Thus,
it becomes a waiting leader
if such event does not occur for sufficiently large $O(m')$ steps,
Otherwise, it increases its level by one or becomes a zombie or a settled agent within $O(m')$ steps. 
\end{proof}

\begin{lemma}
\label{lem:weak_zombie} 
In an execution of $\pal$,
from any time step,
each weak zombie $z$ increases its virtual level or
catches up a leader or a strong zombie
within $O(k)$ steps.
\end{lemma}

\begin{proof}
Suppose that now $z$ is located at a node $v$
and has a virtual level $i$.
By definition of $\pal$,
we must have a path $w=v_1,v_2,\dots,v_s$
such that $s \le k$, $v = v_1$,
$p_{v_i}(v_{i+1}) = a_{v_i}.\last$
for $i=1,2,\dots,s-1$,
and a leader with at least level $i$ is located at $v_s$.
Let $j$ be the smallest integer such that
a leader, a strong zombie, or
an agent with virtual level $i' \ge i+1$
is located at $v_j$.
Then, sub-path $w'$ of $w$ is defined as
$w'=v_1,v_2,\dots,v_j$. 
This sub-path $w'$ changes as time passes.
Zombie $z$ moves forward in $w'$
two times in every four steps 
because it moves in timeslots 2 and 3.
However, a leader moves only once 
or moves and backtracks in every four steps,
while a strong zombie moves only once 
in every four steps.
Therefore, by Lemma \ref{lem:nondecreasing},
the length of the $w'$ 
decreases at least by one 
in every four steps,
from which the lemma follows.
\end{proof}

\begin{lemma}
\label{lem:lmin} 
In an execution of $\Xi=\Xi_{\pal}(C_0)=C_0,C_1,\dots$,
from any time step such that $\lmin < \infty$,
$\lmin$ increases at least by one
in $O(m')$ steps.
\end{lemma}

\begin{proof}
Let $C_t$ be any configuration where $\lmin = i$ holds. 
It suffices to show that the suffix $\Xi'$ of $\Xi$ after $C_t$,
\ie $\Xi'=C_t,C_{t+1},\dots$ reaches a configuration where 
$\lmin \ge i+1$ holds within $O(m')$ steps.
Let $\aset_{\gamma}$ be the set of all weak zombies
with virtual level $i$ that are not located
at a node with a leader or a strong zombie.
No agent in $\aset \setminus \aset_{\gamma}$
becomes an agent in $\aset_{\gamma}$ in $\Xi'$
because all leaders and all strong zombies
must have levels no less than $i$, 
and thus they become weak zombies only after
their virtual levels reach $i+1$.
Therefore, by Lemma \ref{lem:weak_zombie},
$\aset_{\gamma}=\emptyset$ holds within $O(k')$ steps in $\Xi'$,
and $\aset_{\gamma}=\emptyset$ always hold thereafter.
%In $Xi'$, an leader with level $i$ increases its level 
%Thereafter, there is no weak zombie with virtual level $i$.
If there is no agent in $\aset_{\gamma}$,
no waiting leader with at most level $i$
goes back to an active leader
without increasing its level.
Thus, by Lemma \ref{lem:leader_dfs},
%$\Xi'$ reacehs a configuration $C_{t'}$ where 
within $O(m')$ steps,
$\Xi$ reaches a configuration $C_{t'}$
from which there is always no active leader with level $i$.

Let $\Xi''$ be the suffix $C_{t'},C_{t'+1},\dots$ of $\Xi'$.
By definition, in $\Xi''$, 
there is no active leader with level at most $i$.
In addition,
the virtual level of a zombie is $i$
%an agent $a$ has a virtual level $i$
only if it is located at a node with
a strong zombie with level $i$.
Thus, it suffices to show that 
every strong zombie $z$ with level $i$
increases its virtual level at least by one
within $O(m')$ steps.
Suppose that $z$ is located at a node $v$. 
By definition of $\pal$,
we must have a path $w=v_1,v_2,\dots,v_s$
such that $s \le k$, $v = v_1$,
$p_{v_i}(v_{i+1}) = a_{v_i}.\last$
for $i=1,2,\dots,s-1$,
and a leader $l$ is located at $v_s$.
In $\Xi''$, there is no active leader with level $i$.
Thus, if $l$ is active, its level is at least $i+1$.
Even if $l$ is waiting, its level is at least $i$,
and the virtual level of $a_{v_s}$
immediately becomes at least $i+1$
after an active leader or a zombie visits $v_s$. 
Therefore, if none of
the agents $a_{v_1},a_{v_2},\dots,a_{v_{s-1}}$
changes the value of its $\last$
during the next $4s \le 4k$ steps, 
the virtual level of $z$ reaches at least $i+1$.
Thus, suppose that some $a_{v_j}$\ ($1 \le j \le s-1$)
changes the value of its $\last$ during the $4s$ steps.
This yields that an active leader must visit $v_j$
during the period.
Moreover, in $\Xi''$, every active leader has level at least
$i+1$.
Thus, the virtual level of agent $a_{v_j}$ must reach
at least $i+1$ at that time.
Therefore, in any case, 
zombie $z$ increases its virtual level at least by one 
within $O(k)$ steps in $\Xi''$.
\end{proof}

%Let $V_2(t) $ be the set of all nodes in $V$
%at which two or more agents are located
%in time step $t$.
%For any time step $t$ such that $V_2(t) \neq \emptyset$,
%we define $\lmin(t)$ as follows:
%\begin{align*}
%\lmin(t)=\min_{v \in V_2}\max_{a \in \aset(v,t)}a.\level. 
%\end{align*}
%For simplicity, we define $\lmin(t) = \infty$
%if $V_2(t) = \emptyset$.
%Again, we simply write $V_2$ and $\lmin$ for
%$V_2(t)$ and $\lmin(t)$, respectively,
%when time step $t$ is clear from the context.

%\begin{lemma}
%\label{lem:no_two_leaders} 
%There is always at most one leader in $\asetmax(v)$
%for any $v$.
%\end{lemma}
%\begin{proof}
% 
%\end{proof}

%\begin{lemma}
%\label{lem:no_zombie} 
%There is always no zombie in $\asetmax(v)$
%for any $v$.
%\end{lemma}

\begin{theorem}
Algorithm $\pal$
solves a dispersion problem
within $O(m' \log \ell)$ steps.
It uses $O(\log (k+\maxdegree))$ bits of memory space
per agent. 
\end{theorem} 
\begin{proof}
%Remember that $\lmin(t)$ is the minimum virtual level 
%of all active leaders and all zombies in $\aset$
%at step $t$. 
By Lemmas \ref{lem:maxlevel} and \ref{lem:lmin},
$\lmin=\infty$ holds within $O(m' \log \ell)$ steps.
Since $\lmin=\infty$ yields that
no active leader and no zombie exists,
thus all agents are located at different nodes at that time.
Therefore, the theorem immediately follows
from Lemmas \ref{lem:space} and \ref{lem:once}.
\end{proof}

\section{Conclusion}
\label{sec:conclusion}
In this paper, we presented a both time and space efficient algorithm for the dispersion problem. This algorithm does not require any global knowledge. However, we require that all agents compute and move synchronously.
%The assumption of synchronousness is inherent to the proposed algorithm. 
The proposed algorithm inherently requires the synchronous assumption:
active leaders, strong zombies, and weak zombies move in different speeds. We leave open whether both time and space efficient algorithm can be designed for the asynchronous setting.

\bibliographystyle{plain} %参考文献出力スタイル
\bibliography{biblio} %hoge.bibから拡張子を外した名前

\begin{thebibliography}{1}

\bibitem{AM18}
John Augustine and William~K. Moses.
\newblock Dispersion of mobile robots.
\newblock {\em Proceedings of the 19th International Conference on Distributed
  Computing and Networking}, Jan 2018.

\bibitem{GHS83}
Robert~G. Gallager, Pierre~A. Humblet, and Philip~M. Spira.
\newblock A distributed algorithm for minimum-weight spanning trees.
\newblock {\em ACM Transactions on Programming Languages and systems (TOPLAS)},
  5(1):66--77, 1983.

\bibitem{KA19}
Ajay~D Kshemkalyani and Faizan Ali.
\newblock Efficient dispersion of mobile robots on graphs.
\newblock In {\em Proceedings of the 20th International Conference on
  Distributed Computing and Networking}, pages 218--227, 2019.

\bibitem{KMS18}
Ajay~D Kshemkalyani, Anisur~Rahaman Molla, and Gokarna Sharma.
\newblock Efficient dispersion of mobile robots on arbitrary graphs and grids.
\newblock {\em arXiv preprint arXiv:1812.05352}, 2018.

\bibitem{KMS20}
Ajay~D Kshemkalyani, Anisur~Rahaman Molla, and Gokarna Sharma.
\newblock Dispersion of mobile robots in the global communication model.
\newblock In {\em Proceedings of the 21st International Conference on
  Distributed Computing and Networking}, pages 1--10, 2020.

\bibitem{PP99}
Petri{\c{s}}or Panaite and Andrzej Pelc.
\newblock Exploring unknown undirected graphs.
\newblock {\em Journal of Algorithms}, 33(2):281--295, 1999.

\bibitem{PDD+96}
Vyatcheslav~B Priezzhev, Deepak Dhar, Abhishek Dhar, and Supriya Krishnamurthy.
\newblock Eulerian walkers as a model of self-organized criticality.
\newblock {\em Physical Review Letters}, 77(25):5079, 1996.

\bibitem{SBN15}
Yuichi Sudo, Daisuke Baba, Junya Nakamura, Fukuhito Ooshita, Hirotsugu
  Kakugawa, and Toshimitsu Masuzawa.
\newblock A single agent exploration in unknown undirected graphs with
  whiteboards.
\newblock {\em IEICE Transactions on Fundamentals of Electronics,
  Communications and Computer Sciences}, 98(10):2117--2128, 2015.

\bibitem{YWI+03}
Vladimir Yanovski, Israel~A Wagner, and Alfred~M Bruckstein.
\newblock A distributed ant algorithm for efficiently patrolling a network.
\newblock {\em Algorithmica}, 37(3):165--186, 2003.

\end{thebibliography}

%\clearpage 
%\appendix
%\section*{Appendix}
%\section{Omitted Proofs}

\end{document}